\PassOptionsToPackage{english}{babel}
\documentclass[pra,reprint,preprintnumbers,amsmath,amssymb,superscriptaddress,floatfix]{revtex4-2}

\usepackage[T1]{fontenc}
\usepackage[utf8]{inputenc}
\usepackage{times}
\usepackage{graphicx}
\usepackage{hyperref}
\usepackage{url}
\usepackage{amsfonts}
\usepackage{nicefrac}
\usepackage{amsmath}
\usepackage{amsthm}
\usepackage{amssymb}
\usepackage{braket}
\usepackage{xcolor}

\newcommand{\polylog}{\,\textup{polylog}}
\usepackage[capitalise]{cleveref}
\crefname{section}{Sec.}{Secs.}
\Crefname{section}{Section}{Sections}
\crefrangelabelformat{equation}{\textup{(#3#1#4)}--\textup{(#5#2#6)}}
\newtheorem{thm}{Theorem}
\newtheorem{lem}{Lemma}
\newtheorem{conj}{Conjecture}

\begin{document}

\title{Site-by-site quantum state preparation algorithm\\ for preparing vacua of fermionic lattice field theories}

\author{Ali \surname{Hamed Moosavian}}
\affiliation{Joint Center for Quantum Information and Computer Science, NIST/University of Maryland, College Park, Maryland 20742, USA}
\email{muo@protonmail.ch}

\author{James R. Garrison}
\affiliation{Joint Center for Quantum Information and Computer Science, NIST/University of Maryland, College Park, Maryland 20742, USA}
\affiliation{Joint Quantum Institute, NIST/University of Maryland, College Park, Maryland 20742, USA}
\email{jrgarr@umd.edu}

\author{Stephen P. Jordan}
\affiliation{Microsoft Quantum, Redmond, Washington 98052, USA}
\affiliation{University of Maryland Institute for Advanced Computer Studies, College Park, Maryland 20742, USA}
\email{stephen.jordan@microsoft.com}

\begin{abstract}
Answering whether quantum computers can efficiently simulate quantum field theories has both theoretical and practical motivation. From the theoretical point of view, it answers the question of whether a hypothetical computer that utilizes quantum field theory would be more powerful than other quantum computers. From the practical point of view, when reliable quantum computers are eventually built, these algorithms can help us better understand the underlying physics that govern our world.

In the best known quantum algorithms for simulating quantum field theories, the time scaling is dominated by initial state preparation. In this paper, we exclusively focus on state preparation and present a heuristic algorithm that can prepare the vacuum of fermionic systems in more general cases and more efficiently than previous methods. With our method, state preparation is no longer the bottleneck, as its runtime has the same asymptotic scaling with the desired precision as the remainder of the simulation algorithm. We numerically demonstrate the effectiveness of our proposed method for the 1+1 dimensional Gross-Neveu model.

\end{abstract}

\maketitle

\section{Introduction}

One of the main motivations for building quantum computers is to simulate quantum systems efficiently \cite{Feynman1982}, something that is believed to be computationally hard on classical computers \cite{Feynman1985}. Some scattering problems in Quantum Field Theories (QFTs) are BQP-complete \cite{Jordan2018}, and are thus among the most difficult problems that a quantum computer is able to solve. Thus, generic QFTs cannot be simulated in polynomial time on a classical computer unless, of course, quantum computers are actually no more powerful than classical computers (i.e., BQP=BPP in the language of computational complexity theory).
Although Monte Carlo simulations (e.g., for the lattice quantum chromodynamics) can yield static measures like binding energy, doing real-time dynamics for quarks has proven to be difficult \cite{Gattringer2010}.

There are two main approaches to quantum simulation of quantum systems \cite{Preskill2018,Dalmonte2016,Cirac2012}. One approach is to design a system with many quantum degrees of freedom whose dynamics resemble a certain quantum system we want to study. This is called \emph{Analog Quantum Simulation}. Research in this area has been vibrant in the past decade and a half, and possible quantum systems to embed the simulations include but are not limited to ultra cold atoms \cite{Bloch2012}, ion traps \cite{Blatt2012} and Rydberg atoms \cite{Glaetzle2014}. In particular, there are a number of proposals for simulating lattice gauge theories using ultracold atoms in optical lattices \cite{Gonzalez-Cuadra2017,Kasper2017,Banerjee2012,Rico2018,Zohar2015}. Although some of these proposals for analog quantum simulation are quite promising and have been implemented in labs \cite{Bernien2017,Zhang2017}, they have to be handcrafted for each specific problem, and error analysis poses a challenge. The other approach is to use a universal, general purpose \emph{Digital Quantum Computer} to simulate quantum systems. Starting with the pioneering works of \cite{Feynman1982, Deutsch85, Lloyd1996, AbramsLloyd, Zalka}, quantum algorithms for simulating quantum systems using universal digital quantum simulation have become a well-developed area of study.

The known digital quantum algorithms for calculating scattering amplitudes in QFT consist of at least four distinct subroutines \cite{Jordan2012,Jordan2014,HamedMoosavian2018,Klco2018a, Klco2019}. First, they prepare the vacuum state, either by directly preparing the ground state of the interacting theory \cite{HamedMoosavian2018}, or by first preparing the ground state of the noninteracting theory and then adiabatically turning on the interactions \cite{Jordan2014}. The next step is to prepare incoming particle states by adding oscillating terms to the Hamiltonian that couple the vacuum to the desired excited states. Reference \cite{Jordan2012} actually does these two steps in a slightly different manner, by exciting the particles in the noninteracting theory and then adiabatically turning on the interactions. The third step is to let the particles interact and scatter. This is achieved by using an efficient Hamiltonian simulation algorithm, like the ones introduced in Refs.\ \cite{Low2016a,Haah2018,Berry}. The last step is to measure properties of the outgoing particles, such as their locations or momenta. This can be achieved by either adiabatically tuning back to the noninteracting theory or by measuring some local charges with phase estimation \cite{Jordan2014}.

The first two steps together can be thought as the initialization phase of algorithm. In the previous results, initialization has been the bottleneck of the QFT simulation algorithms. For this reason, here we only focus on improving this part of the algorithm, specifically preparation of the vacuum of the interacting theory. The performance of Refs.\ \cite{Jordan2012,Jordan2014} is limited by slow adiabatic transitions in order to avoid exciting extra particles. Reference \cite{HamedMoosavian2018} improves upon the fermionic result by using Matrix Product State properties of one dimensional systems and a classical heuristic algorithm known as the density-matrix renormalization group (DMRG); however, its applications are mostly limited to one dimensional systems and the performance is limited by the classical part of the algorithm \footnote{It is important to note that results like \cite{Schuch2011} that imply gapped one dimensional systems are in the same phase and therefore can be prepared in constant time, do not apply here. In particular they assume that the ratio between correlation length of the system and lattice spacing does not increase with adding more sites to the system. However, in these simulation algorithms one typically assumes system size to be fixed and lattice step decreasing with more sites.}.
In principle, quantum computational power could be used to circumvent this classical bottleneck.

In this paper, we present an efficient method for initial state preparation that is inherently quantum and generalizes to fermionic QFTs in any number of dimensions.
We numerically demonstrate its performance in a 1+1 dimensional fermionic QFT, namely the Gross-Neveu model \cite{Gross1974a}.
In this case in the asymptotic limit of infinite precision for constant system size, the expected performance of this state preparation method scales at least as well with the precision goal $\epsilon$ as the remaining steps in the simulation algorithm. We expect similar performance gains would hold more generally.

In spirit, our algorithm is related to Schwartz, Temme and Verstraete's algorithm for preparing injective \emph{Projected Entangled Pair States} (PEPS) \cite{Schwarz2011}. Their algorithm and ours both grow the system size by adding one site at a time. The main difference is that in their case one needs to know an injective PEPS representation for the state they are preparing, while our algorithm does not. Also, our algorithm performs better with regards to the precision goal $\epsilon$.

The structure of this manuscript is as follows: In \cref{sec: Prelim} we lay out two lemmas and a theorem which are the theoretical foundations of the paper. These are then utilized in \cref{sec: Overview}, which is concerned with explaining our algorithm. \Cref{sec: GrossNeveu} introduces the fermionic Gross-Neveu model as a testbed for our algorithm. Specifically, in \cref{subsec: overview of GrossNeveu} we review the model, and in \cref{subsec: NumericalAnalysis} we provide numerical evidence that our algorithm applies to it. Finally, we conclude in \cref{sec: Conclusion}.

\section{Preliminaries} \label{sec: Prelim}

Aharonov and Ta-Shma in a seminal paper showed that if two states have nonnegligible overlap, with some physically motivated assumptions, one can transition between them in time polynomial in system size \cite{Aharonov2007}. The current paper was in part inspired by their Jagged Adiabatic Path Lemma; however, in the current paper we are using another approach that relies on more modern techniques.

\begin{lem}[Phase estimation with $O((n/m)\log({1/\epsilon}))$ gates]
\label{Lem: Phase estimation}
Given a simulatable Hamiltonian, $H$, that acts on $n$ qubits, and a state, $\ket{\psi}$, and a promised lower bound on the spectral gap, $\Delta(H) > m$, and an estimate for the ground energy, $\tilde{E}$, with a promise that $\left|\tilde{E}-E\right|<\frac m 2$, where $E$ is the actual ground energy; we can check whether $\ket{\psi}$ is the ground state of $H$ or not in runtime proportional to $O((n/m)\log(1/\epsilon))$, where $\epsilon$ is the probability of making a faulty decision.
\end{lem}
\begin{proof}
An $O(\frac{n}{m\epsilon})$ performance can be achieved using the standard phase estimation algorithm \cite{Nielsen2000}. To achieve the $O(\frac{n}{m}\log(1/\epsilon))$ scaling, one can coherently write the output of phase estimation to a number of qubits and then use majority vote and Chernoff bounds to boost the precision \cite{Chernoff1952,Chernoff1981,Nielsen2011}. Specifically, the phase estimation algorithm yields the energy of the state, which we can compare to the given estimate to decide whether it is the ground state or not. In order to achieve linear scaling with $n$, we need to implement a modern and efficient simulation algorithm for local Hamiltonians such as the ones in Ref.~\cite{Childs2019} or Ref.~\cite{Haah2018}.
\end{proof}

\begin{lem}[Mapping overlapping ground states \cite{Yoder2014}]
\label{Lem: Mapping overlapping states}
Given two simulatable Hamiltonians, $H_1$ and $H_2$, with known ground energies, $E_1$ and $E_2$, and a minimum overlap between their ground states, $\left|\braket{g_1|g_2}\right| \ge \eta$, one can get from one ground state to the other with $O\left(\frac{\log(2/\epsilon)}{\eta} \right)$ oracle calls to the phase estimation algorithm on these Hamiltonians, where $\epsilon$ is the precision goal of the algorithm.
\end{lem}
\begin{proof}
This lemma is a direct result of Yoder \emph{et.~al.}'s Grover-esque fixed point quantum search \cite{Yoder2014}, when one replaces the oracle in their paper with our phase estimation from \cref{Lem: Phase estimation}.
\end{proof}

\begin{thm}[Modified Jagged Path Lemma]
\label{Thm: JAPL}
Let us assume $\left\{H_j\right\}_{j=1}^N$ is a sequence of explicit bounded-norm and geometrically local Hamiltonians in a fixed number of dimensions that act on at most $n$ qubits with nonvanishing spectral gaps, $\Delta\left(H_j\right)\ge m_j \ge 0$, where $m_j$ are real positive constants. This means that each Hamiltonian $H_j$ in the sequence is simulatable. Let us also assume we have a priori estimates of the ground energies of these Hamiltonians within accuracy better than half their spectral gap. Then, if the overlap between consecutive ground states, $\ket{g_j}$ and $\ket{g_{j+1}}$, is nonvanishing, $\left| \braket{ g_j|g_{j+1} } \right| \ge \eta > 0 $, there exists an efficient quantum algorithm that can start from $\ket{g_1}$ and output $\ket{ g_N}$ with asymptotic runtime of $O\left(\frac{Nn}{m\eta}\polylog(1/\epsilon)\right)$, where $\epsilon$ is the maximum trace distance of the final state compared to the desired eigenstate and $m = \min\left\{m_j\right\}_{j=1}^N$.
\end{thm}
\begin{proof}
We prove this theorem constructively.
\begin{itemize}
\item We use \cref{Lem: Mapping overlapping states} to transform each ground state in the sequence, $\ket{g_j}$ to $\ket{g_{j+1}}$. Each oracle call is given by the phase estimation from \cref{Lem: Phase estimation}.  Thus, each transformation takes at most $O\left(\frac{n}{m\eta}\polylog(1/\epsilon)\right)$ gates.
\item Because there are $N$ states in the sequence, the overall runtime will be of order $O\left(\frac{Nn}{m\eta}\polylog(1/\epsilon)\right)$.
\end{itemize}
\end{proof}

\section{Overview of the Algorithm} \label{sec: Overview}
\subsection{1+1 dimensions} \label{subsec: 1+1}
There already exist efficient quantum algorithms for preparing the vacuum state of a 1+1 dimensional quantum field theory \cite{HamedMoosavian2018}. Nevertheless, we here use a 1+1 dimensional quantum field theory as a test case for our algorithm, in order to numerically investigate the key unknown quantities determining its performance, namely the state overlaps $\eta$. For now we assume Dirichlet boundary conditions \cite{Costa2019}, although it is easy to generalize the algorithm. The input of the algorithm is a continuous gapped one-dimensional Hamiltonian $H$, system length $L$, maximum error $\epsilon$ and an array of quantum registers initialized in the standard basis, $Q$. Each register consists of qubits which together represent the state of a single site.  The output should be an approximation to the vacuum of $H$ on the array $Q$ that can yield cross sections with precision better than $\epsilon$.

If the system size is much larger than the correlation length of the system, $\chi$, then the inner products should reach a stable value, as the systems are basically unaffected by adding an extra site (\cref{Fig: 1D-lattice}). If the asymptotic value of this inner product is nonzero, it means that if the initial system is larger than a certain size, we can build the vacuum by inductively adding more sites.
\begin{figure}
\includegraphics[scale=0.80]{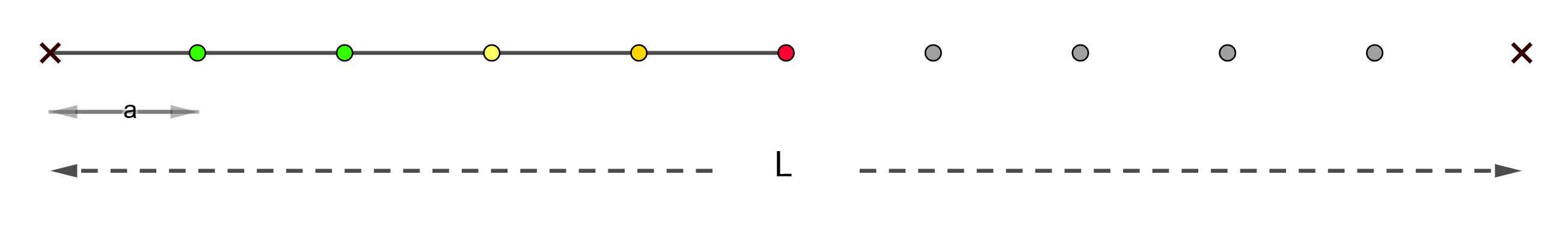}
\caption{(color online) Adding an extra site to a 1D lattice. The crosses represent the boundary (Dirichlet). With a limited correlation length we expect only sites within that distance to be affected by adding an extra site to the system.\label{Fig: 1D-lattice}}
\end{figure}
\subsection{Higher dimensions} \label{subsec: Higher dims}

In higher spatial dimensions, the algorithm is similar to 1+1 dimensions, with the difference that the discretization of the Hamiltonian is more involved and the order of adding extra sites is not uniquely defined.

Suppose the system volume is $V=L_1\times L_2 \times \dots \times L_D$. Then, in the high precision asymptotic limit, the lattice spacing should be $a \sim \epsilon$, where $\epsilon$ represents the maximum relative error of the scattering amplitudes. However, in the high energy limit, the wavelength of the particles would be the deciding factor and $a\sim p^{-1}$, where $p$ represents the momentum of the incoming particles. Similar to the 1+1 dimensional case, one can take $a \sim \epsilon/p$ to respect both asymptotic limits simultaneously \cite{Jordan2014}.

As for the order of adding new sites, one reasonable method is to try to keep it as close to a $D$-dimensional hypercube as possible. \Cref{Fig: 2D-lattice} can be seen as an example of how one can do this in two spatial dimensions, or the side of a three dimensional cube.
\subsection{The algorithm} \label{subsec: algorithm}

In general, our proposal for this state preparation algorithm is as follows. Let us assume our Hamiltonian lies in $D$ spatial dimensions, and its volume is $V=L_1\times L_2 \times \dots \times L_D$. Also, let $\epsilon$ be the precision goal of the entire scattering simulation. Then do the following:
\begin{itemize}
\item Set the lattice spacing, $a$, as $a\propto \epsilon/p$.
\item Properly discretize the Hamiltonian. This means replacing derivatives with finite differences and dealing with discretizing issues such as fermion doubling \cite{Wilson1974,Jordan2014}.
\item Given a boundary condition (e.g.\ Dirichlet), prepare the ground state, $|g_{N_0}\rangle$, of the discretized Hamiltonian with $N_0 \ll N=\frac{V}{a^D}$, i.e.\ a small constant number of sites.
\item Apply a unitary gate (e.g.\ Hadamard) on the rest of the qubits, $|Q_j\rangle \forall j \in \left\{ N_0 +1,\dots,N \right\}$, which is hoped to provide a reasonable overlap between states in the next phase of the algorithm (see \cref{sec:Inner products}).
\end{itemize}
As before, there is only one step in the iterative phase of the algorithm:
\begin{itemize}
\item For every $j\in \left\{ N_0 ,\dots,N-1\right\} $, transform $|g_j\rangle \otimes | Q_{j+1}\rangle$ to $|g_{j+1}\rangle$ by applying \cref{Thm: JAPL}.
\end{itemize}
This yields a runtime of $O\left(\frac{V^2}{a^{2D}\eta}\polylog\left(1/\epsilon\right)\right)$. For the sake of clarity, we will include a conjecture that captures the unproven physical intuition that goes into this algorithm.

\begin{conj}[Overlap of ground states]
Assume a properly discretized massive fermionic QFT that obeys the Wightman axioms, in particular, the energy-momentum spectral condition \cite{Strocchi2004}. Let $\ket{g_j}$ be the ground state of the system with $j$ sites and $\eta$ be defined as $\lim_{j\to \infty}\left|\left(\bra{g_j}\otimes \bra{Q}\right)\ket{g_{j+1}}\right|$, where $\ket{Q}$ is an unentangled state that is present to make the Hilbert spaces compatible. Then there exists $\ket{Q}$ for which $\eta > 0$.
\end{conj}

The value $\eta$ is provably nonzero in many cases, for example if the ground states are described as injective PEPS \cite{Schwarz2011} or if they are topological PEPS \cite{Schwarz2012}.

Some quantum systems will admittedly have ground states that seem to counter the conjecture above. For example, in the AKLT model \cite{Affleck1987a}, the overlap between consecutive ground states is provably zero. However, upon further investigation, one realizes that the AKLT model does not have a single site coarse continuum limit and you need to keep two sites at a time \cite{DelasCuevas2018}. By adding two sites at a time, one can in fact get nonzero and constant overlap between the ground states.

\begin{figure}
\includegraphics[scale=0.22]{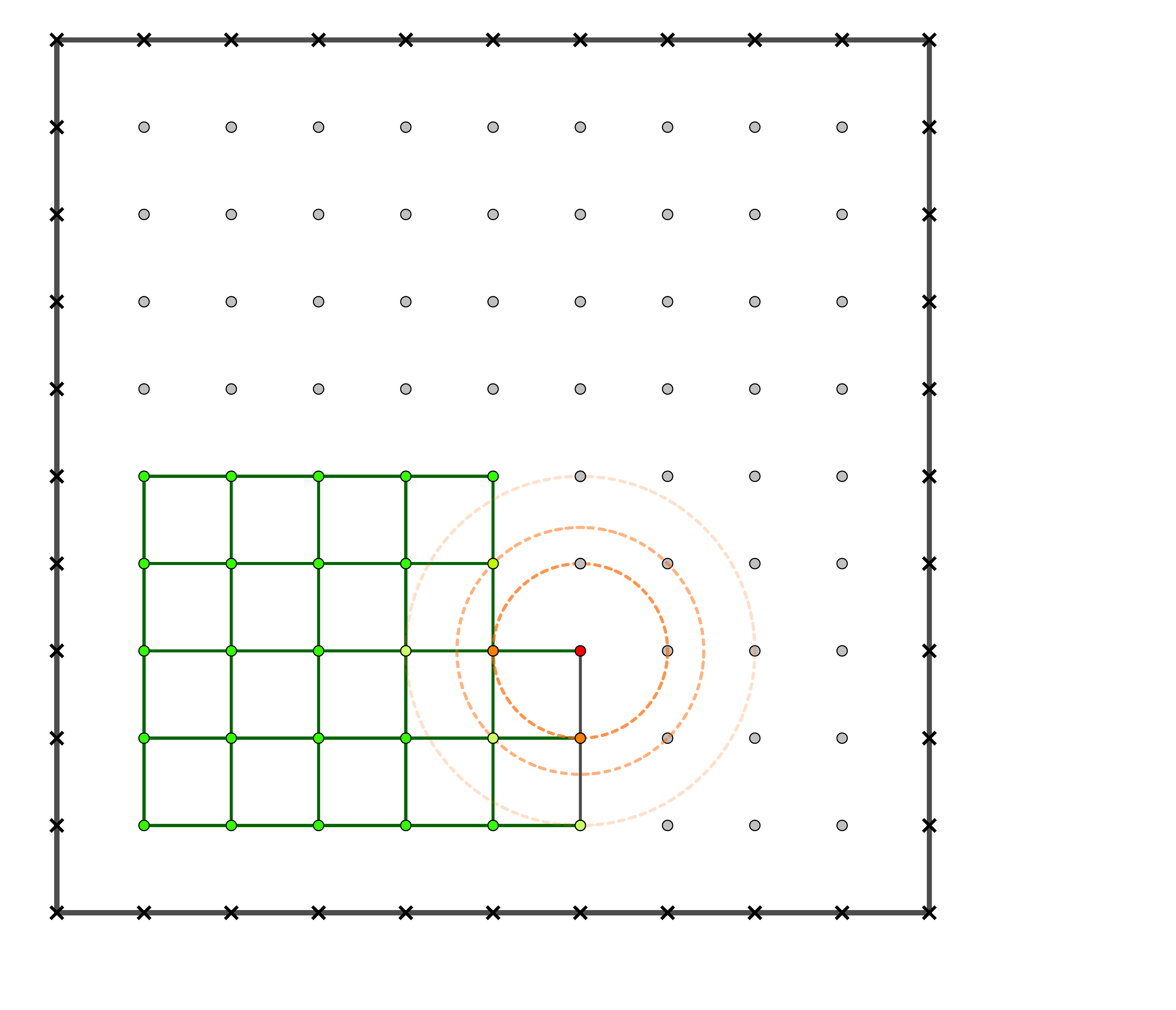}
\caption{(color online) Adding an extra site to a 2D lattice. Because of limited correlation length, only few sites are expected to be affected by the introduction of a new site to the system.\label{Fig: 2D-lattice}}
\end{figure}
\section{Gross-Neveu model} \label{sec: GrossNeveu}
\subsection{Overview of the model} \label{subsec: overview of GrossNeveu}
In this section we will introduce the Gross-Neveu model and use it as a test case for our proposal. The model was introduced in 1974 as a toy model for quantum chromodynamics \cite{Gross1974a}. It is a fermionic QFT that lives in 1+1 spacetime dimensions and exhibits different particle flavors as well as asymptotic freedom. It was originally defined as a massless theory, which has chiral symmetry.  We explicitly break this symmetry by introducing a mass term in the Hamiltonian. The Lagrangian for the massive theory with $\mathcal{N}$ species is given by:
\begin{equation}
 \mathcal{L}=\sum_{j=1}^{\mathcal{N}}\bar{\psi}_j\left(i\gamma^{\mu}\partial_{\mu}-m\right)\psi_j+\frac{g^2}{2}\left(\sum_{j=1}^{\mathcal{N}}\bar{\psi}_j\psi_j\right)^2\, ,
\end{equation}
where $g$ represents the interaction strength, $\gamma^{\mu}$ are two-dimensional representations of the Dirac field, $\bar{\psi}=\psi^{\dagger}\gamma^0$, and each field $\psi_j$ has two components \cite{Jordan2014}. We use the Majorana representation for the $\gamma$ matrices, where they are explicitly written as:
\begin{equation}
 \gamma^0 = i\left( \begin{array}{cc}
             0 & -1 \\
             1 & 0
            \end{array}\right)\, ,
\end{equation}
\begin{equation}
\gamma^1 = -i\left(\begin{array}{cc}
                   0 & 1 \\
                   1 & 0
                   \end{array}
\right) \, .
\end{equation}

For simulation purposes it is more convenient to work with the equivalent Hamiltonian formalism. Additionally, to simulate the scattering process on a digital quantum computer we need to discretize the model and put it on a lattice. Discretizing the model and putting it on a lattice introduces extra fermions; this is known as the fermion doubling problem \cite{Nielsen1981,Nielsen1981a}.  These extra fermions can be handled via different methods such as Wilson fermions \cite{Wilson1974}, Kogut-Susskind staggered fermions \cite{Kogut1975, Banks1976, Susskind1977} or domain wall fermions \cite{Kaplan1992}. For instance, if we had periodic boundary conditions and we had wanted to utilize Wilson fermions, we would have had to add an extra term to the Hamiltonian that decouples the extra fermions from the ground state (\cref{Fig: Wilson}). The full Hamiltonian of the system after discretizing would then be \cite{HamedMoosavian2018}:
\begin{equation}
H = H_0 + H_g + H_W \,,
\end{equation}
where
\begin{widetext}
\begin{align}
H_0 & = \sum_{x \in \Omega} a \sum_{j=1}^{\mathcal{N}} \sum_{\alpha,\beta\in \{0,1\}}\bar{\psi}_{j,\alpha}(x)
\left[ -i \gamma_{\alpha\beta}^1 \frac{\psi_{j,\beta}(x + a)
- \psi_{j,\beta}(x-a)}{2a} + m_0 \delta_{\alpha,\beta}\psi_{j,\beta}(x) \right] \label{h0} \,, \\
H_g & = -\frac{g_0^2}{2} \sum_{x \in \Omega} a \bigg( \sum_{j=1}^{\mathcal{N}}\sum_{\alpha\in \{0,1\}}
\bar{\psi}_{j,\alpha} (x) \psi_{j,\alpha}(x) \bigg)^2 \label{hg} \,, \\
H_W & = \sum_{x \in \Omega} a \sum_{j=1}^{\mathcal{N}} \sum_{\alpha\in \{0,1\}} \left[ - \frac{r}{2a}
\bar{\psi}_{j,\alpha}(x) \left( \psi_{j,\alpha}(x+a) - 2 \psi_{j,\alpha}(x) + \psi_{j,\alpha}(x-a) \right)
\right] \,.
\label{hw}
\end{align}
\end{widetext}
Here, $H_0$ represents the noninteracting term of the Hamiltonian, $H_g$ represents the interaction term, and $H_W$ is the Wilson term. The summation variable $j \in \{1,2,\dots,\mathcal{N}\}$ indicates the fermion species, and $0 < r
\leq 1$ is called the Wilson parameter. $H$ is spatially local in the sense
that it consists only of single-site and nearest-neighbor terms on the
lattice.

If one wants to simulate the continuum limit of the Gross-Neveu model, they should eliminate the doubled fermions through some mathematical procedure. However, in our minimal approach for a numerical example, it suffices to note that the extra particles are not necessarily a problem. In our test example, the doubled fermions can be thought as extra flavors of fermions.

The $\mathcal{N}=1$ case of the massive Gross-Neveu model, which we will be using to check our proposal, is equivalent to the massive Thirring model, which in turn can be solved analytically using Bethe ansatz \cite{Haldane1982,Mandelstam1975,Coleman1975,Korepin1979,Okwamoto1983,
Zamolodchikov1979}. Although Bethe ansatz is a powerful tool, it does not work for all systems, and in this specific case the solutions are rather complicated. Instead, we focus here on more general numerical approaches, which can in principle work for arbitrary $\mathcal{N}$. Specifically, we rely on a DMRG algorithm \cite{Schollwock2011a} to classically calculate the ground state as a Matrix Product State.  The DMRG code we developed is written in Julia \cite{Bezanson2017} and is available online \cite{HamedMoosavian2019}.

\begin{figure}
\includegraphics[scale=0.27]{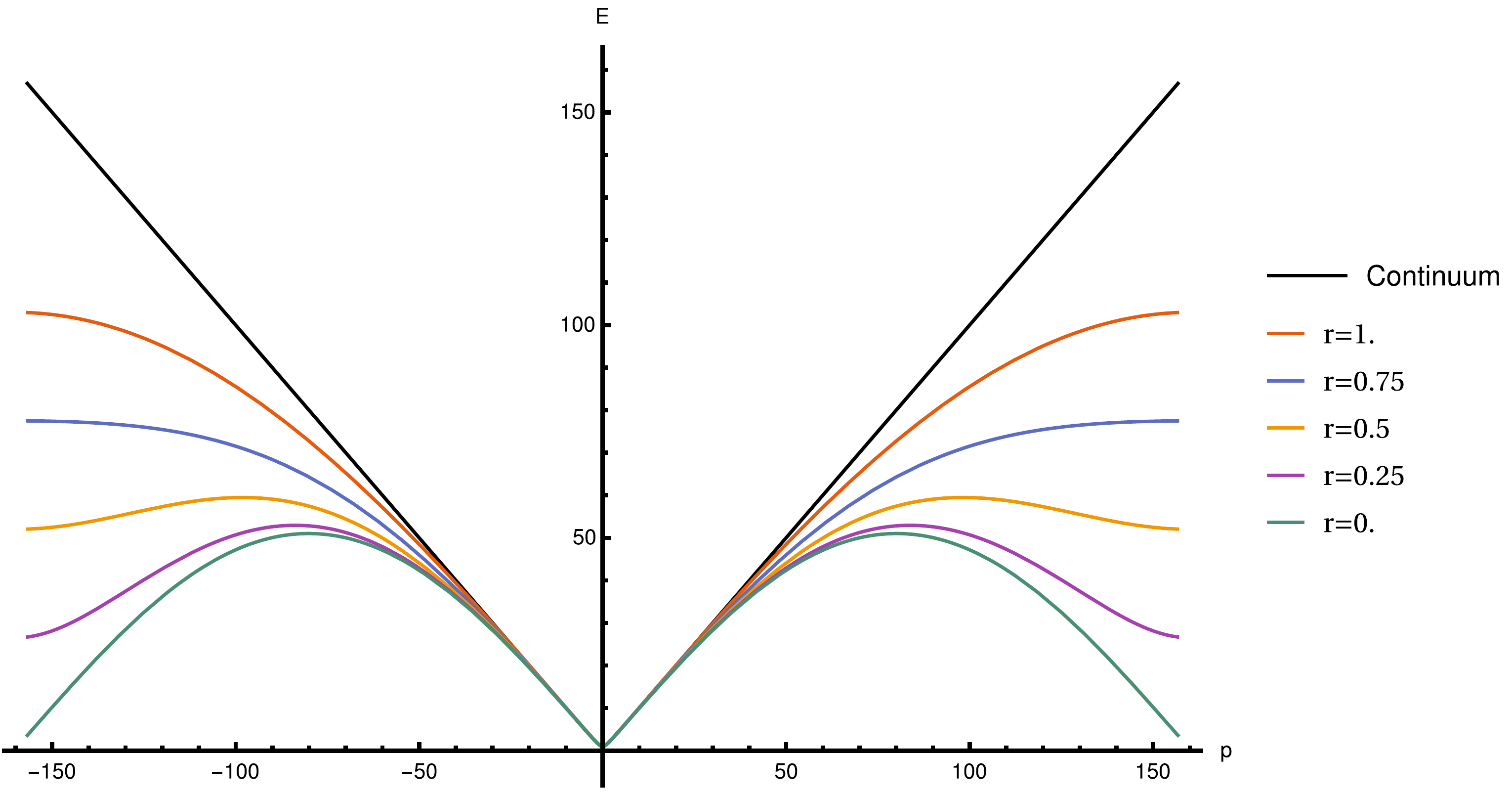}
\caption{(color online) The dispersion relation of the non-interacting theory for different values of the Wilson parameter, $r$, compared to the continuum limit. In this plot $m_{0}$ is set
to be $m_{0}=1$ and the descretized system is set to have unit length with $N=50$ sites. The dispersion relation in this case is given by \cite{Jordan2014}: $E_{m_{0}}^{(a)}(\mathbf{p})=\sqrt{\left(m_{0}+\frac{2r}{a}\sin^{2}\left(\frac{\left|\mathbf{p}\right|a}{2}\right)\right)^{2}+\frac{1}{a^{2}}\sin^{2}\left(\left|\mathbf{p}\right|a\right)}\, .$
\label{Fig: Wilson}}
\end{figure}

\subsection{Numerical analysis and diagrams} \label{subsec: NumericalAnalysis}
Ideally, if one had access to a sufficiently advanced quantum computer, one might first choose the desired simulation parameters and then use them to determine the system size that is necessary for accurate simulation at that point.
However, with limited classical computational power, we can only verify our proposal for reasonably small system sizes.
Therefore, we aim to find a range of parameters, $m_{0},g_{0}$, that can be simulated accurately on a system with $\sim 50$ sites.  Specifically, the parameter regimes we choose must yield ground states with correlation lengths that are simultaneously much smaller than our simulation size and much larger than the lattice spacing.

\subsubsection{Mass renormalization} \label{subsubsec: Mass Renormalization}
If the interaction strength is set to zero (i.e., we are working in the free theory), then it is straightforward to calculate the two-point correlation functions in the continuum limit.

\begin{equation}
\psi(x)=\int \frac{dp}{2\pi}\frac{1}{\sqrt{2E_p}}\left(a(p)u(p)e^{-ip \cdot x}+b^{\dagger}(p)v(p)e^{ip \cdot x}\right)\ ,
\end{equation}
where $E_p=\sqrt{p^2+m_0^2}$, $u(p)=\left(
    \begin{array}{c}
      \sqrt{E_p-p} \\
      i\sqrt{E_p+p}
    \end{array}
  \right)$, and $v(p)=\left(
      \begin{array}{c}
        \sqrt{E_p-p} \\
        -i\sqrt{E_p+p}
      \end{array}
\right)$. $a(p)$ and $b^{\dagger}(p)$ are creation and annihilation operators. Then the two-point correlation function can be calculated as:
\begin{align}
\langle 0 |\psi_0(x)\bar\psi_0(y)|0\rangle &= \int\frac{dp\, dq}{\left(2\pi\right)^2}\frac{1}{2\sqrt{2E_pE_q}}e^{i(qy-px)} \nonumber \\
& \qquad \times \sqrt{E_p-p}\sqrt{E_q+q}\langle 0|a(p)a^{\dagger}(q)|0\rangle \\
 &= \int_{-\infty}^{\infty}\frac{dp}{2\pi}\frac{m_0}{2\sqrt{p^2+m_0^2}}e^{ip(x-y)} \\
 &= \frac{m_0}{2\pi} K_0\left(m_0\left|x-y\right|\right) \ ,
\end{align}
where $K_0$ is the modified Bessel function of the second kind. We expect the two-point correlation functions to keep such a form even in the discretized and interacting case, i.e.,
\begin{equation}
\langle 0 |\psi_0(x)\bar\psi_0(x+\Delta x)|0\rangle \propto K_0\left( \frac{\Delta x}{\chi(m)}\right) + O(\epsilon) \label{Eq: TwoPointbehavior}\ ,
\end{equation}
where $\chi(m) \propto 1/m$ is the correlation length, which is generally inversely proportional to the renormalized mass. Asymptotically, $K_0(\zeta)$ behaves like an exponentially decaying function in the limit $\zeta\to\infty$ \cite{Hastings2005}:
\begin{equation}
K_0(\zeta) \sim \sqrt{\frac{\pi}{2\zeta}}e^{-\zeta}\left(1-\frac{1}{8\zeta}+\frac{9}{128\zeta^2}+O\left(\frac{1}{\zeta^3}\right)\right)\ .
\end{equation}
Going forward, we use an exponentially decaying form to numerically calculate the renormalized mass from the two-point correlation functions.

\begin{figure*}
\includegraphics[scale=0.55]{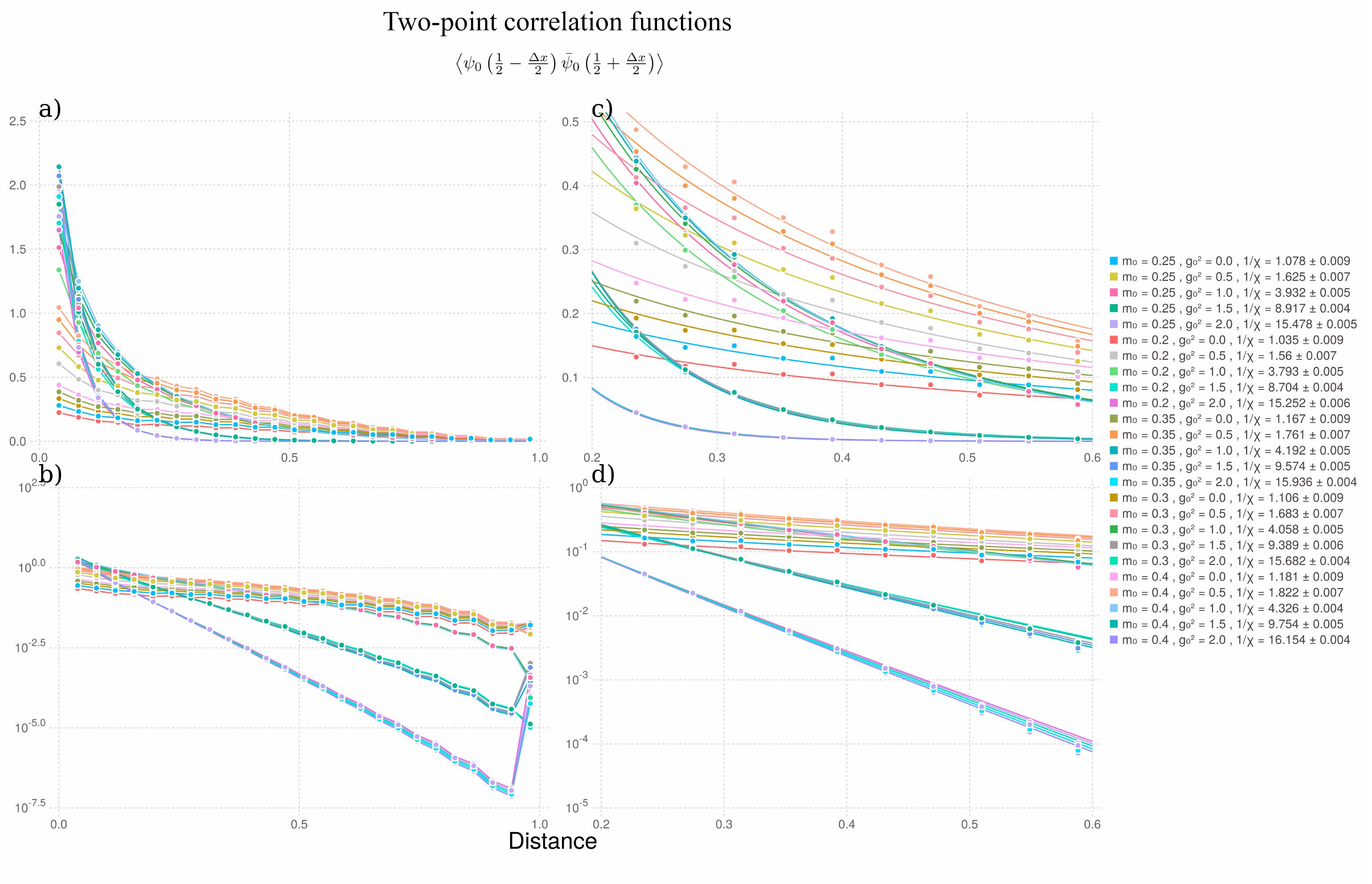}
\caption{(color online) A sample of two-point correlation functions calculated for different set of $m_0$ and $g_0^2$ parameters. In subplots a) and b) we can see the entirety of the two correlation functions calculated over the span of distances. In subplots c) and d) we have only kept a range of distances in the middle. Also, the curves in these two subplots are the result of fitting $bK_0(\frac{\Delta x}{\chi})$ to the date points. The legend on the right shows parameters $m_0$ and $g_0^2$ for each set of data points as well as the inverse of the calculated correlation length, $\frac{1}{\chi}$.  \label{Fig: TwoPoint}}
\end{figure*}

\subsubsection{Bare mass and interaction strength} \label{subsubsec: parameters}
We investigate a range of values for bare mass, $m_0$, and interaction strength, $g_0$, specifically looking for the sets of parameters that yield correlation lengths that are much longer than our lattice spacing and at the same time much smaller than the system size. Because we have set the length of the system to be $1$ (in units of inverse energy), a reasonable correlation length should be around $\sqrt{\frac 1 {51}} \approx \frac 1 {7.14}$, or about $7$ lattice spacings. The goal of the rest of this subsection is not to pinpoint the parameters that nail such a correlation length, but rather find values that yield viable correlation lengths that ensure the calculations in the following subsections are valid.

After some preliminary calculations it seems that a good range of parameters that yield reasonable correlation lengths would be $g_0^2 \in \left[ 0,2.0\right]$ and $m_0 \in \left[0.2,0.4\right]$. We calculate the two-point correlation function, $\left\langle \psi_0\left( \frac{1}{2} - \frac{\Delta x}{2} \right)\bar\psi_0\left(\frac{1}{2}+\frac{\Delta x}{2}\right)\right\rangle,\forall x\in\{a,2a,\dots,\frac N 2 a\}$, of the system for a uniform distribution of parameters in that range. (Some of these two-point correlation functions can be seen in \cref{Fig: TwoPoint}.) \Cref{Eq: TwoPointbehavior} should hold for distances much larger than the lattice spacing. Therefore, ideally we are interested in the long range behavior of these correlation functions. However, at very long distances because of boundary effects and limited machine precision, our numerics deviate from \cref{Eq: TwoPointbehavior}. In order to avoid these issues, we hand pick a range of distances that are much larger than the lattice spacing and at the same time are far enough from the edge of the system. This range of distances shows the least amount of deviations. We determine the correlation length, $\chi$, by finding the best set of values, $b$ and $\chi$, that fit the data with $bK_0\left(\frac{\Delta x}{\chi}\right)$.

Based on the results in \cref{Fig: TwoPoint}, we deem parameters $\left(m_0 = 0.2 , g_0^2 = 1.5\right)$ and $\left(m_0 = 0.4 , g_0^2 = 1.0\right)$ to have correlation lengths suitable for further numerical calculations at the desired system size of $N=50$.

\subsubsection{Inner products} \label{sec:Inner products}
Now that we have found a set of reasonable parameters, let us look at the inner product between the ground states of systems with different numbers of lattice sites. In order to make the inner product well defined, we need to add unentangled extra sites to the smaller system so the Hilbert spaces will be the same. The extra site we add in our numerical analysis is a uniform superposition over the standard basis. In the case of $\mathcal{N}=1$, we need two qubits per site after mapping the fermionic system to qubits using a Jordan-Wigner \cite{Jordan1928,Batista2000} transformation (see Ref.\ \cite{HamedMoosavian2018} for a detailed explanation of this mapping for the Gross-Neveu model). The state of the extra site would be:
\begin{equation}
 \frac{1}{2}\sum_{j=0}^3|\bar{j}\rangle\, ,
\end{equation}
where the bar in $\bar j$, means it is written in base $2$ \footnote{Note that in most cases, especially systems that have conserved quantum numbers, this would not yield the maximum inner product. For example, in our case of the massive Gross-Neveu model, by choosing the extra site to be in the $(\ket{01}+\ket{10})/ \sqrt 2$ state, one can improve the inner products by a factor of $\sqrt2$. Choosing the best state for the extra site requires some knowledge about the symmetries of the system, which in many cases can be difficult to know before doing deep analysis of the system. Therefore, to be conservative, we decided to use the more generic uniform superposition for the added site.}. Now that the inner product is well defined, we increase the number of lattice sites, one at a time, and calculate the inner product between these ground states. As can be seen in \cref{Fig: Inner}, the inner products rapidly converge to a positive constant. This shows that our conjecture works for these sets of parameters of the Gross-Neveu model. Therefore, assuming we can classically estimate ground energies, our algorithm can be used to prepare their ground states.

\begin{figure}
 \includegraphics[scale = 0.45]{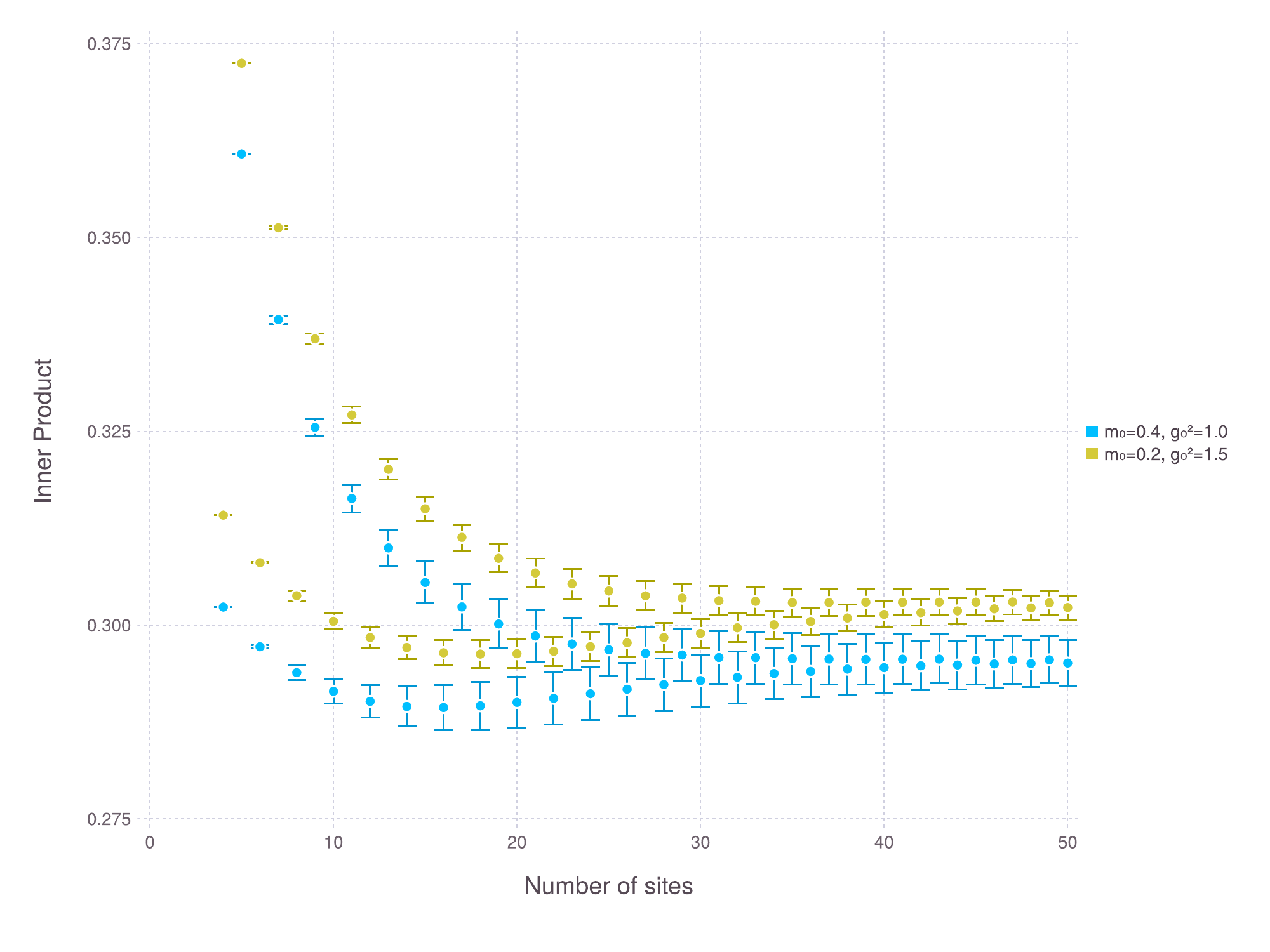}
 \centering
 \caption{(color online) Inner products between systems with different numbers of lattice sites. The inner products are calculated for $\left(m_0 = 0.2 , g_0^2 = 1.5\right)$ and $\left(m_0 = 0.4 , g_0^2 = 1.0\right)$.\label{Fig: Inner}}
\end{figure}

With a back of the envelope calculation we expect the asymptotic value of the overlap, $\eta=\lim_{j\to \infty}\left(\langle g_j|\otimes \langle Q_{j+1}|\right)|g_{j+1}\rangle$, to be $\eta \propto e^{-\nicefrac{\chi}{a}}$ in one spatial dimension. (In $D$ spatial dimensions we expect this to be $\eta \propto \exp\left(-\frac{\chi^D}{a^D}\right)$, as more sites are affected by the introduction of a new site to the system.) However, what we observe in \cref{Fig: Inner} shows inner products of surprisingly large magnitude and mild dependence on correlation length. We were surprised by this result, though it is good news for our algorithm, and we hope in the future to investigate additional lattice models to find out how generally it holds.

\subsubsection{Predicting the energy} \label{subsubsec: energy prediction}
There is still one condition from \cref{Thm: JAPL} that needs to be satisfied before we can prepare the vacuum of the Gross-Neveu model using that theorem: We should be able to predict the ground state energy with accuracy better than half of the gap. The gap is equal to the renormalized mass. The ground state energy is expected to grow almost linearly with the number of lattice sites with minute corrections from the Casimir effect. Therefore, the ground state energy can be approximated as \cite{Bordag2001}:

\begin{align}
E_g(L) &= C_0 + C_1 L + C_2 \sum_{h=1}^{\infty}\frac{1}{h^2}K_2\left(C_3 h L\right)\\
&\approx c_0 + c_1 L + \frac{c_2}{L} + \frac{c_3}{L^2} + \frac{c_4}{L^3} + O\left(\frac{1}{L^4}\right)\ ,
\end{align}
where $L$ represents the size of the system or number of sites, $K_2$ is the modified Bessel function of the second kind, and $c_j$ and $C_j$ are constant real numbers that depend on the geometry of the system \cite{Farina2006}. In \cref{Fig: FitDifference} we use all of the previous energy data points to predict the next ground state energy. As  is evident from \cref{Fig: FitDifference}, after some system size our prediction of the energy is well within half of the gap size.

\begin{figure*}
\includegraphics[scale = 0.85]{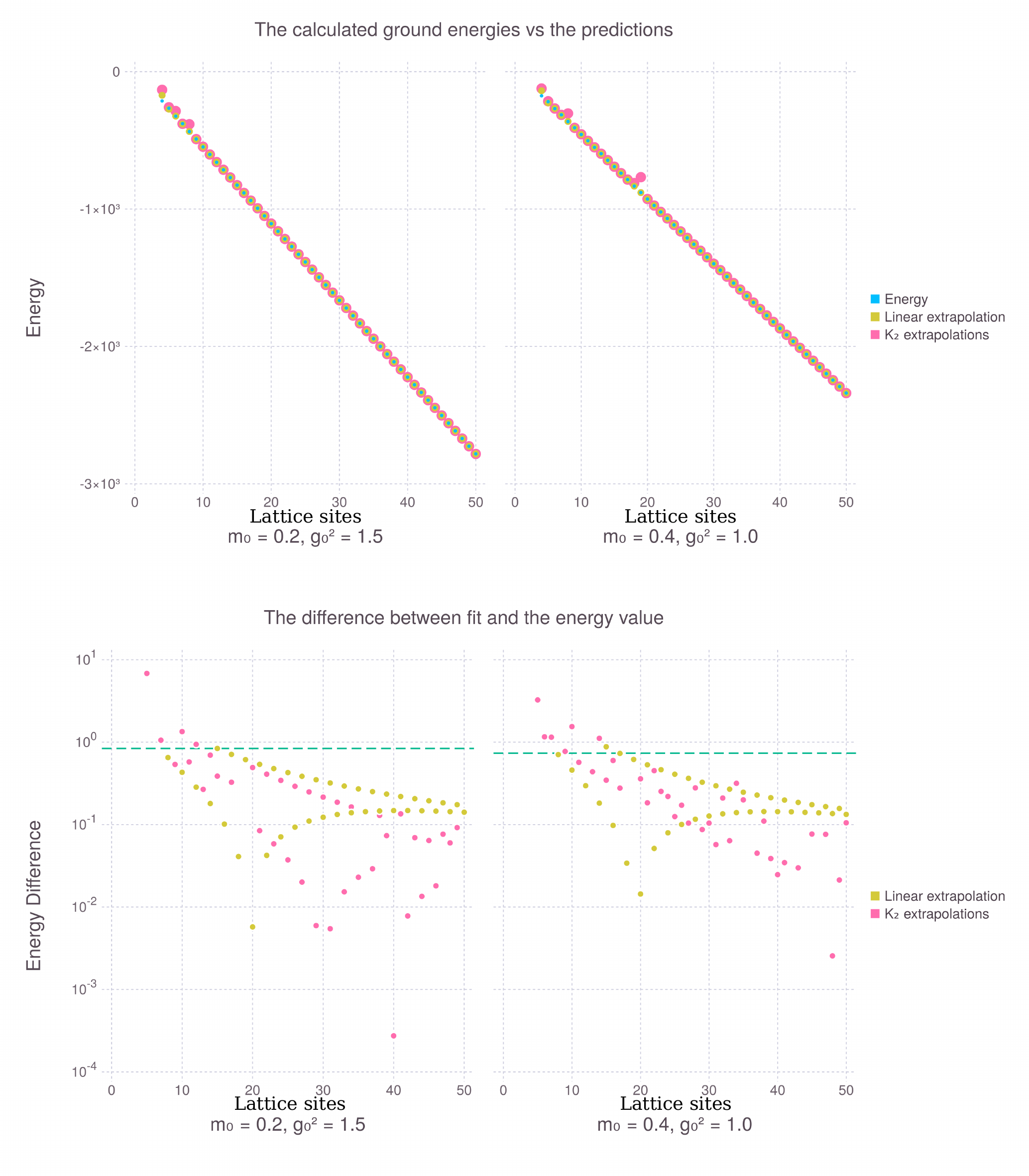}
\centering
\caption{
(color online) For each data point in this figure, we analyze the energies from previous system sizes by fitting a curve to the values, predicting an energy for the next system size. The top plots show the ground state energies next to the predicted energies, and the bottom plots display the difference between our predictions and the ground state energies. In the asymptotic limit of small systems, the Casimir effect can be approximated by $C_2\sum_{h=1}^{10}\frac{1}{h^2}K_2\left(C_3 h L\right)$, where $L$ is the system size and $C_2$ and $C_3$ are some geometrical constants \cite{Farina2006}. As one can see, this theoretical argument fits our numerical data too. The mustard dots represent a linear fit, $f(L)=c_0 + c_1 L$ and the pink dots present the difference between the data and $f(L) = C_0 + C_1 L + C_2 \sum_{h=1}^{10}\frac{1}{h^2}K_2\left(C_3 h L\right)$. The dashed teal lines represent the minimum precision required, which is half of the spectral gap, i.e.\ the renormalized mass. Clearly a linear fit with some estimation for energy density is well within the gap bounds; the higher order fits show the consistency of our calculations. \label{Fig: FitDifference}}
\end{figure*}

\subsubsection{Error analysis for numerical calculations} \label{subsubsec: Error analysis}
For numerical calculations, the quantity we use to measure the precision of the ground state is the following \cite{Schollwock2011a}:
\begin{equation}
\epsilon=\frac{\left|\left\langle \psi\left|H^{2}\right|\psi\right\rangle \right|-\left|\left\langle \psi\left|H\right|\psi\right\rangle \right|^{2}}{\left|\left\langle \psi\left|H\right|\psi\right\rangle \right|^{2}}\, .\label{eq:Def_epsilon}
\end{equation}
The DMRG algorithm stops whenever $\epsilon < \epsilon_{\text{goal}}$ or the bond dimension reaches a maximum value, where $\epsilon_{\text{goal}}$ is the predefined precision goal. We need to know, given a value of $\epsilon$, what the distance between the result of DMRG and the actual ground state is. Let us assume, to the first nonzero order in $\delta$:
\begin{equation}
|\psi\rangle=|g\rangle+\delta|g^{\perp}\rangle\ ,
\end{equation}
where $ H|g\rangle=E_{0}|g\rangle$ represents the ground state, $|g^{\perp}\rangle$ is a state orthogonal to the ground state, and $\delta$ is a small value. We have:
\begin{alignat}{2}
\left|\braket{\psi|H|\psi} \right| & =\left|E_{0}+\delta^{2}\braket{g^{\perp}|H|g^{\perp}} \right|&&\le\left|E_{0}+\delta^{2}\braket{H} \right|\, ,\\
\left|\braket{\psi|H^{2}|\psi} \right| & =\left|E_{0}^{2}+\delta^{2}\braket{g^{\perp}|H^{2}|g^{\perp}} \right|&&\le\left|E_{0}^{2}+\delta^{2}\braket{H^{2}} \right|\, .
\end{alignat}
Substituting these values into \cref{eq:Def_epsilon}, we get:
\begin{equation}
\begin{aligned}
\Rightarrow\epsilon & \approx\delta^{2}\left|\frac{\left\langle H^{2}\right\rangle -2E_{0}\left\langle H\right\rangle }{E_{0}^{2}}\right|\\
& = \delta^{2}\left|\left(\frac{E_\text{max}}{E_{0}}\right)^{2}-2\frac{E_\text{max}}{E_{0}}\right|\\
& = \delta^2\left|\kappa^2-2\kappa\right|\, ,
\end{aligned}
\end{equation}
\begin{equation}
\delta \lessapprox \frac{\sqrt{\epsilon}}{\kappa^2-2\kappa} < \sqrt{\epsilon} \ ,
\end{equation}
where $\kappa$ is the condition number of $H$. $\delta$ is upper bounded by $\sqrt{\epsilon}$, and we use this upper bound in our analysis throughout the paper and in \cref{Fig: TwoPoint} for the error bars.

\section{Conclusion and open problems} \label{sec: Conclusion}

In this paper, we have introduced an algorithm that can help with ground state preparation of fermionic QFTs. In particular, our algorithm performs better than state of the art algorithms \cite{HamedMoosavian2018}, and it can be generalized to any number of spatial dimensions. Specifically, initialization is no longer the bottleneck of fermionic QFT simulation, as its runtime has the same asymptotic scaling as the rest of the algorithm. Overall, it is a humble step towards answering whether the entirety of the Standard Model can be simulated on a universal digital quantum computer.

It is important to note that although our conjecture about $\eta$ can be rigorously proven in a number of cases such as systems where the ground state of the theory has a known topological PEPS representation \cite{Schwarz2011}, whether our algorithm will work for every gapped fermionic system is an open problem.

Also, we believe that the bosonic case needs further investigation and a similar algorithm might work in that case too. One difference one needs to be aware of in that case is that bosonic statistics allow several bosons to occupy the same site and this will necessitate the introduction of a cut-off.

Perhaps a harder open problem to consider is how to generalize these state preparation algorithms to critical systems that lack a mass gap.

\begin{acknowledgments}

  We thank Robin Kothari, Ali Lavasani, Guang Hao Low, Yannick Meurice, John Preskill, Seth Whitsitt, and members of the theoretical Quarks, Hadrons and Nuclei group at the Maryland Center for Fundamental Physics for helpful discussions.  This work was supported in part by the U.S. Department of Energy (DOE) under Award Number DE-SC0019139. JRG acknowledges support by the DoE ASCR FAR-QC (award No.\ DE-SC0020312), DoE BES QIS program (award No.\ DE-SC0019449), NSF PFCQC program, DoE ASCR Quantum Testbed Pathfinder program (award No.\ DE-SC0019040), AFOSR, NSF PFC at JQI, ARO MURI, and ARL CDQI\@. JRG performed his work in part at the Aspen Center for Physics, which is supported by National Science Foundation grant PHY-1607611.
  The authors acknowledge the University of Maryland supercomputing resources made available for conducting the research reported in this paper. This work also used the Extreme Science and Engineering Discovery Environment (XSEDE), which is supported by National Science Foundation grant number ACI-1548562.  Specifically, it used the Bridges system, which is supported by NSF award number ACI-1445606, at the Pittsburgh Supercomputing Center (PSC)~\cite{xsede, bridges}.
\end{acknowledgments}

\bibliography{references}
\end{document}